\documentclass[sigconf,screen,nonacm]{acmart}
\usepackage{amsmath}
 \usepackage{amsthm}
\usepackage{mathtools}
\usepackage{tikz}
\usetikzlibrary{arrows}
\usetikzlibrary{positioning} 
\usepackage{tikz-cd}
\usepackage{tikz-timing}
\usetikzlibrary{automata, positioning, arrows}
\usepackage{nicefrac}
\usepackage{pict2e}
 \usepackage{enumitem}
 \usepackage{stmaryrd}

\newcommand*{\Defeq}{\,\coloneqq\,}
\newcommand{\Def}{\Defeq}

\newcommand*{\Comp}{\,\circ\,}

 \newcommand{\Filter}[1]{\,\triangleright\,}

\newcommand*{\Nat}{\mathbb{N}} 
\newcommand*{\Natinf}{\Nat\cup\{+\infty\}}

\newcommand*{\Bool}{\mathbb{B}}
\newcommand*{\Real}{\mathbb{R}}

\newcommand*{\Inter}{\,\cap\,}
\newcommand*{\Union}{\,\cup\,}

\newcommand*{\Max}[2]{\mathrm{max}({#1},{#2})}
\newcommand*{\Min}[2]{\mathrm{min}({#1},{#2})}

\newcommand*{\Setc}[2]{\{ {#1}\,\mid\, {#2} \}}   
\newcommand*{\Powerset}[1]{\mathcal{P}(#1)}

\newcommand*{\CB}[1]{{\mathrm{CB}({#1})}}  
\newcommand*{\Compact}[1]{{\mathrm{Comp}({#1})}}  

\newcommand*{\V}[1]{{#1}}
\newcommand*{\Product}[2]{(\Pi {#1})\,{#2}}
\newcommand*{\Hausdorff}[1]{\mathcal{H}_{#1}}
\newcommand*{\Hd}[3]{\mathcal{H}_{#1}({#2},{#3})}

\newcommand*{\fix}[1]{{\mathrm{fix}}({#1})}
\renewcommand*{\comment}[1]{}

\newcommand*{\Par}{\otimes}

\newcommand*{\Magic}{\mathrm{magic}}
\newcommand*{\Abort}{\mathrm{abort}}
\newcommand*{\Angelic}{\mathrm{\sqcup}}
\newcommand*{\Demonic}{\mathrm{\sqcap}}

\newcommand{\Ord}[1]{v({#1})}
\newcommand{\Norm}[1]{|{#1}|}

\newcommand{\Poly}[2]{{#1}[{#2}]}

\newcommand{\Power}[2]{{#1}\llbracket {#2}\rrbracket}

\newcommand*{\Imp}{\,\Rightarrow\,}

\newcommand*{\Iff}{\,\iff\,}

\newcommand*{\Forall}[2]{(\forall {#1})\,{#2}}

\newcommand*{\Cross}{\times}

\newcommand{\Doubleplus}{+\kern-1.3ex+\kern0.8ex}

\newcommand{\Stream}[1]{{#1}^{\omega}}
\newcommand{\StreamN}[2]{{{#1}^{\omega, {#2}}}}
\newcommand{\Head}[1]{\mathrm{hd}({#1})}
\newcommand{\Tail}[1]{\mathrm{tl}({#1})}

\newcommand*{\Prefix}[2]{{#1}\!\!\restriction_{#2}}

\newcommand{\ITE}[3]{\mathrm{{#1} \,?\, {#2}  \,:\, {#3}}}
\newcommand{\LOOP}[1]{{#1}^+}

\newcommand{\SP}[1]{{\mathrm{sp}}_{#1}}
\newcommand{\WP}[1]{{\mathrm{wp}}_{#1}}

\newcommand{\BO}[2]{\mathrm{B}({#1}, {#2})}        
\newcommand{\BC}[2]{\mathrm{B}[{#1}, {#2}]}  

\newcommand{\Count}[2]{\mathrm{\#}_{#1}{#2}}

\AtBeginDocument{%
  \providecommand\BibTeX{{%
    \normalfont B\kern-0.5em{\scshape i\kern-0.25em b}\kern-0.8em\TeX}}}

\begin{document}
\title{Solving Causal Stream Inclusions}
\author{Harald Ruess}
\email{harald.ruess@entalus.com}
\orcid{0000-0002-1405-2990}
\affiliation{%
  \institution{Entalus Computer Science Labs}
  \streetaddress{2071 Gulf of Mexico Drive}
  \city{Longboat Key}
  \state{FL}
  \country{USA}
  \postcode{34228-3202}
}

\renewcommand{\shortauthors}{}

\begin{abstract}
  We study solutions to systems of stream inclusions of the form $f \in T(f)$, where the nondeterministic transformer $T$ on $\omega$-infinite streams is assumed to be {\em causal} in the sense that elements in output streams are determined by a finite prefix of inputs. 
  We first establish a correspondence between logic-based {\em causality} 
  and metric-based {\em contraction}\@. 
  Based on this {\em causality-contraction connection} we then apply fixpoint principles to the spherically complete ultrametric space of streams to construct solutions of stream inclusions.
The underlying fixpoint iterations induce {\em fixpoint induction principles} to reason about these solutions.
In addition, the fixpoint approximation provides an {\em anytime algorithm} with which finite prefixes of solutions can be calculated.
These developments are illustrated
for some central concepts of system design.
\end{abstract}

\maketitle

\section{Introduction}\label{sec:intro}
We consider existence, uniqueness, approximations, and reasoning principles
for solutions (fixpoints) of the {\em stream inclusions}
    \begin{align}\label{stream.inclusions}
     f \in T(f)\mbox{\@,}
    \end{align}
where $f \Defeq (f_1, \ldots, f_n)$, 
for $n \in \Nat$\@, is a vector of 
infinite {\em streams} $f_i$ of values, 
and the vector- and multivalued {\em stream transformer} $T$ is {\em causal} in that every element in the output stream is completely determined by a finite history of inputs.
These stream inclusions are ubiquitous in 
computer science
and other fields of knowledge such as biology, economics, and artificial intelligence to model the evolution of 
dynamic systems under uncertainty.

Typical sources of uncertainty are lack of epistemic knowledge about the behavior of the system under consideration as well as modeling artifacts such as underspecification, whereby a system is said to be nondeterministic and a system specification is underspecified 
if for some input several outputs are admitted.
Underspecified functions map streams onto sets of streams, and
they are isomorphic to relations on streams. 

Consider, for example, mixing two Boolean input streams $f, g \in 2^{\omega}$ such that its output stream $h \in 2^{\omega}$ contains as many $1$s as the two input streams combined; i.e.,  $h \in \mathit{mix}(f, g)$ if and only if  $\Count{1}{h} = \Count{1}{f} + \Count{1}{g}$\@.
Now, a feedback loop is easily modeled by $f \in \mathit{mix}(f, g)$, specified by the equation $\Count{1}{f} = \Count{1}{g} + \Count{1}{f}$\@.
Validity of expected properties such as 
$\Count{1}{g} = 0 \Imp \Count{1}{f} = 0$, however, can only be shown by assuming $\mathit{mix}$ to be causal~\cite{broy2023}\@.
Other prominent examples of nondeterministic functions include {\em merging} of two streams in such a way that all of any infinite input stream is absorbed, and the related problem of {\em fair scheduling}~\cite{park2005concurrency}\@. 
In general, any relation between input and output streams is specified by such nondeterministic stream transformer.
Now, two mutually dependent nondeterministic stream transformers, say $T_1$ and $T_2$ give rise to a system
   \begin{align*}
     f_1 &\in T_1(f_1, f_2) \\
     f_2 &\in T_2(f_1, f_2)
   \end{align*}
with two stream inclusions, which can easily be recast in the vector-valued form~(\ref{stream.inclusions})\@.

If the stream transformer $T$ 
is deterministic then the stream inclusion~(\ref{stream.inclusions}) reduces to a system of stream equalities $f = T(f)$\@.
Interacting systems of deterministic stream processors, for instance, are traditionally modeled as the least fixpoint of this equality, where $T$ is a Scott-continuous stream transformer in a complete partial order~\cite{gilles1974semantics}\@.
But such a denotational semantics does not 
extend na{\"i}vely to unbounded nondeterminism~\cite{back1980semantics,staples1985fixpoint,pratt1984pomset,broy1986theory}\@.
One option are extensions of complete partial orders into power domains~\cite{plotkin1976powerdomain,smyth1978power}, but there is no partial order 
on sets of streams that can be used as the approximation ordering for defining sets of least fixpoints of multi-valued functions on streams~\cite{broy2023}\@.

We take a different approach by investigating under which conditions fixpoints exist for causal stream transformers with unconstrained nondeterminism, which is naturally modeled with multivalued maps.
Here we distinguish between {\em weakly} and {\em strongly causal} stream transformers,
whereby the latter notion also implies a strict, bounded delay for which the outputs are determined.

The main results for solving stream inclusions of the form~(\ref{stream.inclusions})
are as follows.
\begin{itemize}
\item 
For strongly causal stream transformers $T$ with nonempty, compact codomains, solutions $f \in T(f)$ are contained in the (unique) fixpoint $F = \SP{T}(F)$ for the {\em strongest post} of $T$~(Theorem~\ref{thm:main0})\@.
This latter fixpoint is obtained as the limit of a Picard iteration with an explicit quantitative bound on the approximation of each iteration. 
\item  $f \in T(f)$ has a solution if $T$ is strongly causal, not the constant map to the empty set, and all the codomains of $T$ are closed in the topology induced by the prefix metric on streams~(Theorem~\ref{thm:main1})\@. In addition, we identify a slightly stronger condition than strong causality to establish the uniqueness of solutions of the stream inclusion~(\ref{stream.inclusions})\@. 
A fixpoint induction principle is derived to reason about these solutions.
\item  $f \in T(f)$ has a solution if $T$ is strongly causal and all codomains of $T$ are nonempty and compact~(Theorem~\ref{thm:main2}). 
\item  If the stream transformer $T$ is weakly causal then either $f \in T(f)$ has a solution or there exists a ball of streams with positive radius on which the distance, as measured by the prefix metric on streams, between input stream and corresponding output sets with respect to $T$ is, in a sense to be made precise, invariant~ (Theorem~\ref{thm:main3})\@. 
\end{itemize}

These developments are based on the correspondence between the logic-based notion of causality and the quantitative concept of contraction (and nonexpansion)\@.
More precisely, we show that a vector- and multivalued stream transformer is weakly causal if and only if it is {\em nonexpansive}, and it is strongly causal if and only if it is {\em contractive} for the prefix metric.
This approach is motivated by recent work of Broy~\cite{broy2023} on stream-based system design calculus for strongly causal transformers, which is based on Banach's fixpoint principle to solve stream equations on deterministic stream transformers.
We generalize this observation by
    (1) modeling nondeterministic and mutually dependent system components 
            as conjunctions of causal stream inclusions,
    (2) establishing the equivalence of vector- and multivalued causal stream transformers with contraction in a spherically complete ultrametric space of streams based on the prefix distance of streams, and
    (3) applying multivalued fixpoint principles in this ultrametric stream space
          to obtain solutions of causal stream inclusions. 
 
 First, Section~\ref{sec:preliminaries} explains some basic concepts of distance and topological spaces that are used in this paper.
 The prefix distance between two streams is then measured in Section~\ref{sec:streams}
 using the longest common prefix. 
 This results in a spherically 
 complete ultrametric space of streams. 
 The underlying concept of the prefix distance derives from Cantor sets, and it generalizes to arbitrary stream products and also to sets of streams via the induced Hausdorff distance.
 Next, Section~\ref{sec:transformers} discusses the familiar concepts of {\em stream transformers} 
 along with some horizontal and vertical composition operators.
 The focus in Section~\ref{sec:causality} is on {\em causal} stream transformers, whose outputs are 
 determined by a finite history of inputs. 
 For instance, causality is preserved under composition and refinement of stream transformers. 

In Section~\ref{sec:contractions} we develop a metric-based characterization of causal stream transformers.
More specifically, we show that a nondeterministic 
stream transformer is causal if and only if it is contractive with respect to the given prefix ultrametric on streams. 
This correspondence extends to a notion of Lipschitz contraction based on the Hausdorff distance between sets of streams,
since stream transformers with nonempty compact images are Lipschitz contractive if and only if they are causal.

Since the induced Hausdorff metric is a complete metric space on the set of nonempty compact sets, we obtain unique fixpoints in Section~\ref{sec:fixpoints} for causal transformers such as the {\em weakest pre} and the {\em strongest post} set transformers\@.
Moreover, there is a linearly, strictly decreasing upper bound of the prefix distance between this fixpoint and its approximation by the underlying Picard iteration, 
which suggests an {\em anytime} approximation algorithm.
Furthermore, we formulate, in the spirit of Park's lemma, induction principles for fixpoints of set transformers, and we show that the set of fixpoints of a multivalued stream transformer $T$ is contained in the fixpoint of the strongest post transformer for $T$\@.

Additional fixpoint results in Section~\ref{sec:fixpoints} construct solutions for the stream inclusion~(\ref{stream.inclusions}) 
by showing that every strongly causal vector- and multivalued stream transformer has a fixpoint as long as it is not the constant map to the empty set and its codomain is restricted to closed sets only.  We also identify a slightly 
stronger condition for which this fixpoint is unique. The underlying Picard iteration enables us to derive
an induction principle for reasoning about these fixpoints.
In Section~\ref{sec:fixpoints} we also state some immediate consequences of the correspondence of causality and Lipschitz contractivity
for fixpoints of weakly causal maps.
Section~\ref{sec:remarks}  discusses further consequences of the causality-contraction connection, and
Section~\ref{sec:conclusions} concludes with a discussion on the relevance of these developments for the principled design of systems.

\section{Preliminaries}\label{sec:preliminaries}

We summarize a hodgepodge of concepts and notation for topological and metric spaces as they will be used in the remainder of this paper, with the intention of making the developments of this paper as self-contained as possible. 
Readers who are familiar with these fundamental concepts can proceed to the next section.

For a metric space $(M, d)$ 
the
sets  $\BO{x}{r} \Defeq \Setc{y}{d(x, y) < r}$ and $\BC{x}{r} \Defeq \Setc{y}{d(x, y) \leq r}$
are called the {\em open} and {\em closed balls} of {\em center} $x$ and {\em radius} $r$, respectively\@.
The family of open balls forms a base of neighborhoods for a uniquely determined Hausdorff topology on $M$, which is the {\em topology induced by $d$} (on $M$)\@.
Open, closed, bounded, (dis)connected, convex, totally bounded (precompact), and compact sets are defined with respect to the metric-induced topology.
The set of nonempty, closed, and bounded subsets of $M$, in particular, is 
denoted by $\CB{M}$,  and $\Compact{M}$ is the set of nonempty compact subsets of $M$\@. 

The {\em distance} $d(a, B)$ of an element $a \in M$ to a nonempty set $B \subseteq M$ is defined by
     $d(a, B) = \inf_{b \in B}{d(a, b)}$\@. 
Clearly, $d(a, b) = d(a, \{b\})$ for all $a, b \in M$\@. 
For a bounded metric space $(M, d)$ the {\em Hausdorff distance} $\Hausdorff{d}(A, B)$ between two nonempty sets
$A, B \subseteq M$ measures the "longest path" to get from $A$ to $B$, or vice versa, from $B$ to $A$\@. 
     \begin{align}
       \Hausdorff{d}(A, B) &\Defeq \max{(\sup_{x \in A}{(d(x, B))},
                              \sup_{y \in B}{(d(y, A))})}\mbox{\@.}
     \end{align}
Clearly, all suprema exist for a bounded  metric $d$, 
and $d(a, b) = \Hausdorff{d}(\{a\}, \{b\})$\@. 
Furthermore, if $A$, $B$ are closed then their Hausdorff distance
$\Hausdorff{d}(A, B)$ is finite, and $(\CB{M}, \Hausdorff{d})$ is a metric space\@.
 
A sequence $(x_k)_{k \in \Nat}$ in $M$  is {\em Cauchy} if and only if 
for all $\varepsilon > 0$ there exists $N \in \Nat$ such that
$d(x_n, x_m) < \varepsilon$ for all $n, m \geq N$\@. 
Such a sequence $x_k$ {\em converges }to $x^* \in M$ if and only if
for each neighborhood $U$ of $x$ there exists $N \in \Nat$ such that $x_n \in U$ for all $n \geq N$\@.
Now, the space $M$ is {\em Cauchy complete} if every Cauchy sequence converges to some $x^* \in M$\@.  
Equivalently, $M$ is Cauchy complete if and only if the intersection of nested sequences
of closed balls whose radius approaches to $0$ are nonempty\@.

A map $T: M \to M$ (in a metric space $M$) is {\em contractive}
if there exists a constant $l$ with $0 < l < 1$ such that
$d(T(x), T(y)) \leq l \cdot d(x, y)$ for all $x, y \in M$\@.
Traditionally, {\em Banach's contraction principle} establishes that
in a Cauchy complete metric space $(M, d)$ there is a unique fixpoint
for every contracting $T: M \to M$~\cite{kelley1955general}\@. 
Starting with the {\em Picard iteration} $x_{n+1} = T(x_n)$ 
for $n\in\Nat$ with $x_0$ arbitrary, one concludes from the contraction property 
with Lipschitz constant $0 \leq l <1$, that $d(x_{n+1}, x_n) \leq l\cdot d(x_n, x_{n-1})$,
and,  therefore, $d(x_{n+1}, x_n) \leq \nicefrac{l^n}{1 - l} \cdot d(x_0, x_1)$\@. 
From this, one concludes that the sequence $(x_n)_{n \in \Nat}$ is Cauchy, 
and, for completeness, that its limit $x^* \in M$ is a fixpoint. 
This fixpoint is unique, since we get  $x^* = y^*$ for any two fixpoints from
 $d(x^*, y^*) = d(T(x^*), T(y^*)) \leq l \cdot d(x^*, y^*)$\@. 

A map $T: M \to M$ is said to be {\em shrinking}
if $d(T(x), T(y)) < d(x, y)$ for all $x, y \in M$\@.
A shrinking map $T$ need not have a 
fixpoint in a complete metric space.

An {\em ultrametric} space $(M, d)$ is a metric space with the {\em strong triangle inequality}, 
for all $x, y, z \in M$\
    \begin{align}\label{strong.triangle.inequality}
    d(x, y) &\leq \Max{d(x, y)}{d(y, z)}\mbox{\@.}
    \end{align}
 
As a consequence of~(\ref{strong.triangle.inequality}) the {\em isosceles triangle principle}
       \begin{align}
       d(x, z) &= \Max{d(x, y)}{d(y, z)}
       \end{align}
holds whenever $d(x, y) \neq d(y, z)$\@. 
Further immediate consequences of the strong triangle inequality are:
(1) every point inside a ball is its center, that is, if 
$d(x, y) < r$ then $\BO{x}{r} = \BO{y}{r}$,
(2) all balls of strictly positive radius $r$ are {\em clopen}, that is both open and closed,
(3) if two balls are not disjoint then one is included in the other,
(4) the {\em distance} $d(B_1, B_2) \Defeq \inf_{x\in B_1, y \in B_2}d(x, y)$ of two disjoint nonempty balls $B_1$,$B_2$ is obtained as the distance of two arbitrarily chosen elements $x \in B_1$, $y \in B_2$\@. 
(4) ultrametric spaces are {\em totally disconnected},
that is, every superset of a singleton set is disconnected,
and 
(5) a sequence $(x_k)_{k \in \Nat}$ in an ultrametric space is Cauchy
if and only if $\lim_{k\to \infty}d(x_{k+1}, x_k) = 0$\@. 
We will also make use of a generalized strong triangle 
inequality 
for ultrametric distances. 
  \begin{align}
      d(a, C) &\leq \max{(d(a, b), d(b, C))}\mbox{\@,} 
  \end{align}
for all $a, b \in M$ and $\emptyset \neq C \subseteq M$\@.
%
\begin{proposition}
If $(M, d)$ is ultrametric then $(\CB{M}, \Hausdorff{d})$ is also ultrametric.
\end{proposition}

An ultrametric space $(M, d)$ is {\em spherically complete} if
$\Inter_{{B \in \mathcal C}} B \neq \emptyset$ for every chain ${\mathcal C}$
of balls $B_0 \supseteq B_1 \supseteq B_2 \supseteq \ldots$\@.
From this definition it is clear that spherical completeness implies Cauchy completeness. 
\begin{example}\label{ex:finite-ultrametric}
Let $M = \{\alpha,\beta,\gamma,\delta\}$ with 
$d(x,x) = 0$ for all $x \in M$,  
$d(\alpha,\beta) = d(\gamma,\delta) = \nicefrac{1}{2}$, 
$d(\alpha,\gamma) = d(\alpha,\delta) = d(\beta,\gamma) = d(\beta,\delta) = 1$, and
$d(y,x) = d(x,y)$ for all $x, y \in M$\@. 
Then $(M, d)$ is a spherically complete ultrametric space.
\end{example}
\begin{proposition}[(\cite{dieudonne2011foundations}, p. 59]
If $(M, d)$ is an ultrametric space then the ultrametric space
$(\Compact{M}, \Hausdorff{d})$ is 
Cauchy complete.
\end{proposition}
 
For a spherically complete ultrametric space $M$, every
strictly contracting $T:M \to M$
has a unique fixpoint in $M$~(\cite{petalas1993fixed}, Theorem 1)\@.
Moreover, if $T: M \to M$ is {\em nonexpansive}, that is
$d(T(x), T(y)) \leq d(x, y)$ for all $x, y \in M$, then
either $T$ has at least one fixpoint 
or there exists a ball $B$ of radius
$r > 0$ such that $T: B \to B$ and for which $d(u, T(u)) = r$ for each $u \in B$ ~(\cite{petalas1993fixed}, Theorem 2)\@.
Such a ball $B$ is said to be
{\em minimal $T$-invariant}\@.


\section{Streams}\label{sec:streams}
An $A$-valued {\em stream} in $\Stream{A} \Defeq \Nat \to A$, for a given nonempty set $A$ of {\em values}, 
is an infinite sequence $(a_k)_{k \in \Nat}$ with  $a_k \in A$\@.
Depending on the application context, streams are also referred to as discrete streams or signals, 
$\omega$-streams, $\omega$-sequences,  or $\omega$-words\@.
The {\em generating function}~\cite{charalambides2018enumerative} of a stream  $(a_k)_{k \in \Nat}$ is a
{\em formal power series}
   \begin{align}
   \sum_{k\in\Nat} a_k X^k
  \end{align}
in the {\em indefinite} $X$\@. 
These power series are {\em formal} as, in the algebraic view, 
the symbol $X$ is not being instantiated and there is no notion of convergence. 
We call $a_k$ the {\em coefficient} of $X^k$, and
the set of formal power series with coefficients in $A$ is 
denoted by $\Power{A}{X}$\@.
We also write $[X^k]f$ for the coeffient of $X^k$ in the formal power series $f$\@. 
Now, $\Head{f} \Defeq [X^0]f$ and 
$\Tail{f}$ is the unique stream such that $f = \Head{f} + X \cdot \Tail{f}$\@. 
A {\em polynomial} in $\Poly{A}{X}$ of degree $d \in \Nat$ is a formal power series $f$ which is {\em dull}, that is $[X^d]f \neq 0$ and $[X^n]f = 0$ for all $n > d$\@. 
For the one-to-one relationship between streams $\Stream{A}$ and formal power series $\Power{A}{X}$
we use these notions interchangeably. 
Streams are added componentwise and they are multiplied by {\em discrete convolution}\@.
   \begin{align}  \label{add}
       (\sum_{k\in\Nat} a_k X^k) + (\sum_{k\in\Nat} b_k X^k)           &\Defeq   \sum_{k\in\Nat}  (a_k + b_k) X^k \\  \label{multiply}
        (\sum_{k\in\Nat} a_k X^k) \cdot (\sum_{k\in \Nat} b_k X^k)          &\Defeq   \sum_{k \in\Nat}  (\sum_{i = 0}^k a_i b_{k-i}) X^k 
   \end{align}
With these operations and $A$ an (integrity) ring, $\Stream{A}$ becomes a commutative
(integrity) ring with zero element $0 \Defeq \Bar{0}$ and multiplicative identity $1 \Defeq \Bar{1}$\@.
Hereby $\Bar{a} \Defeq (a + \sum_{k \geq 1} 0X^k)$, for $a \in A$, is the injective and homomorphic {\em embedding} of the ring $A$ into the ring $\Power{A}{X}$ of formal power series.
Similarly, the ring of polynomials in $\Poly{A}{X}$ is injectively and homomorphically embedded in $\Power{A}{X}$ as dull formal power series\@. 

If $A$ is a field, then $\Stream{A} \simeq \Power{A}{X}$ is a {\em principal ideal domain} with the ideal 
$(X) = X\cdot\Stream{A}$ the only nonzero maximal ideal. 
Moreover, for $A$ a field or a  division ring, $(\Stream{A}, +, (a\cdot)_{a \in K})$
is a {\em linear space},
whereby the {\em dot product} $(a\cdot f)$, for $f \in \Stream{A}$, is defined by $\Bar{a} \cdot f$\@. 
Unless stated otherwise, we assume the values $A$ in $\Stream{A}$ to be a field\@. 

The multiplicative inverse $f^{-1}$ for $f \in \Stream{A}$ exists (in which case it
is unique) if and only if $[X^0]f \neq 0$\@.
We also write $\nicefrac{f}{g}$ instead of $f \cdot g^{-1}$\@.
In particular, the multiplicative inverse 
of $X$ does not exist in $\Stream{A}$\@. 
As a consequence, $\Stream{A}$ is not a field even when $A$ is a field. 

The {\em valuation} 
$\Ord{f}$ of a stream $f \in \Stream{A}$ is the minimal $k \in \Nat$ 
such that $[X^k]f \neq 0$, if any exists; otherwise $\Ord{f} \Defeq \infty$\@. 
By construction,
 \begin{align}\label{valuation}
     \Norm{f} \Defeq 2^{-\Ord{f}}\mbox{\@,}
   \end{align}
with $2^{-\infty} \Defeq 0$\@, 
is the {\em non-Archimedean absolute value} on $\Stream{A}$
induced by the valuation $\Ord{.}$\@~\cite{neukirch2013algebraic}\@. 
Notice that stream valuation trivially has the {\em non-Archimedean} property $\Norm{1 + \cdots + 1} \leq \Norm{1} = 1$, where $\leq$ is the total order on streams, as induced by the valuation $\Ord{.}$\@. 

The {\em prefix} distance between streams $f$ and $g$ is measured in terms of the 
longest common prefix:  the longer the common prefix, the closer a pair of streams.
Via the distance function
   \begin{align}\label{def:distance}
        d(f, g) &\Defeq \Norm{f - g}
   \end{align}
the set $\Stream{A}$ is a metric with a 
discrete {\em set of values} 
     $\Delta_d \Defeq  
      \Setc{2^{-n}}{n \in \Nat \cup \{\infty\}}
     $\@.
In fact, $d$ is an {\em ultrametric}, since, by construction,
the {\em strong triangle inequality}
   \begin{align}\label{def:strengthened.triangle}
   d(f, h) &\leq max(d(f,g), d(g, h))\mbox{\@.}
   \end{align}
holds for all streams $f$, $g$, $h$\@. 
\begin{proposition}
Both addition (\ref{add}) and multiplication (\ref{multiply}) of streams are continuous with respect to the topology induced by the prefix metric~$d$\@. 
\end{proposition}
The topology induced by the prefix metric $d$ is identical to 
the product topology $A^{\Nat}$, where each copy of $A$ is the discrete topology. 
Therefore Tychonoff's theorem applies, and 
$\Stream{A}$ is compact if and only if $A$ is finite. 


The following fact can easily be checked from the definition of spherically completeness, which requires that the intersection of non-increasing sequences of non-empty balls in $\Stream{A}$ is nonempty.
\begin{lemma}\label{lem:streams.spherically.complete}
   $(\Stream{A}, d)$ of streams is spherically complete.
\end{lemma}
As a consequence, $(\Stream{A}, d)$ is also Cauchy complete, and,
indeed, $\Stream{A} \simeq \Power{A}{X}$ is the Cauchy completion of the polynomials $\Poly{A}{X}$ for the prefix metric $d$\@.
Let ${\mathcal I}$ be a nonempty index set,
$A_{\iota}$ a set of values for each index $\iota \in {\mathcal I}$, 
and $\Product{\iota \in {\mathcal I}}{\Stream{A_{\iota}}}$ the set of ${\mathcal I}$-indexed product of formal power series.
The valuation 
    \begin{align}\label{valuation.product}
    |\bar{f}|_{{\mathcal I}} &\Defeq \sup_{\iota \in {\mathcal I}}     |f_{\iota}|
    \end{align}
for products $\bar{f}$ of the form $ (f_{\iota})_{\iota \in {\mathcal I}}$ is 
just the supremum of the valuation of its components.

In particular, the valuation of finite dimensional products of the form
$(f_1, \ldots, f_n)$, for $n \in \Nat$, is
$2^{-k}$, where $k \in \Nat \Union \{\infty\}$ is the maximal position such that all prefixes $\Prefix{f_i}{k}$, for $i = 1, \ldots, n$, only contain zeros. 
Likewise, when interpreting $\bar{f} \in \Product{\iota \in {\mathcal I}}{\Stream{A_{\iota}}}$
as a (dependent) function with domain ${\mathcal I}$
and codomains $\Stream{A_{\iota}}$ for $\iota \in {\mathcal I}$, 
then the valuation $|\bar{f}|_{\mathcal I}$ is obtained
as
$\sup_{g \in \bar{f}({\mathcal I})} |g|$\@. 

A metric on the ${\mathcal I}$-indexed product space 
$\Product{\iota \in {\mathcal I}}{\Stream{A_{\iota}}}$ is induced
by the valuation (\ref{valuation.product})\@.
   \begin{align}
    d_{\mathcal I}(\bar{f}, \bar{g}) \Defeq |\bar{f} - \bar{g}|_{\mathcal I} 
   \end{align}
\begin{lemma}\label{lem:product.ultrametric.complete}
For nonempty ${\mathcal I}$ and $A_\iota$ for each $\iota \in {\mathcal I}$,
the space
$(\Product{\iota \in {\mathcal I}}{\Stream{A_{\iota}}, d_{{\mathcal I}}})$ is ultrametric and spherically complete.
\end{lemma}
Since $d_{{\mathcal I}}$ specializes for a singleton index set to the distance $d$ on streams, we usually drop the subindex ${\mathcal I}$\@.

\section{Transformers}\label{sec:transformers}

Stream transformers are the basic building blocks for modeling (discrete) dynamical systems, say $T$, with a vector $\V{x}$ of  $n$ input
streams  and a vector $\V{y}$ of $m$ output streams\@.
We write $\StreamN{A}{n}$ for the set $(\Stream{A})^n$ of
$n$-dimensional vectors of $A$-valued streams. 

\begin{definition}[Stream Transformers]\label{def:stream.transformers}
  A {\em stream transformer}
  is a vector- and multivalued map $T: \StreamN{A}{n} \to \Powerset{{\StreamN{B}{m}}}$, where $n, m \geq 1$\@.
  If $|T(f)| = 1$ (componentwise) for all $f \in (\StreamN{A}{n})$ then $T$ 
 is a {\em deterministic} stream transformer; otherwise the
 stream transformer is said to be {\em nondeterministic}\@. 
 For a deterministic stream transformer we also write $T: {\StreamN{A}{n}} \to {\StreamN{B}{m}}$\@. 
\end{definition}
The direct image of a stream transformer $T$ with respect to a set $F$ of streams is denoted by $T(F)$\@. 

The restriction to vector-valued stream transformers is mainly motivated by notational convenience, as most of the developments in this paper generalize to transformers with heterogenous coefficient sets of the form
   $$(\Stream{A_1} \Cross \ldots \Cross \Stream{A_n}) \to 
           \Powerset{\Stream{B_1} \Cross \ldots \Cross \Stream{B_m}}\mbox{\@,}$$
and also to infinite products. 
Vector-valued stream transformers ${\StreamN{A}{n}} \to \Powerset{{\StreamN{B}{m}}}$ may be identified with stream transformers in $\Stream{(A^n)} \to \Powerset{\Stream{(B^m)}}$, because of the one-to-one relationship of $\Stream{(A^n)}$ with $\StreamN{A}{n}$, and
$\Stream{(B^m)}$ with $\StreamN{B}{m}$\@.
If $A$ ($B$) is a field then $A^n$ ($B^m$) can be made into a field in the usual way.


\begin{example}~\label{ex:stream-circuits}
{\em Stream circuits}~\cite{rutten2008rational} 
are clocked, hardware-like finite structures, 
which are obtained by (finite) compositions of the
(rational) stream transformers 
of the form ($r \in \Real$)
   \begin{align*}
   M_r(z) &\Defeq r \cdot z\\
   A(z_1, z_2) &\Defeq z_1 + z_2 \\
   C(z) &\Defeq (z, z)^\tau\\
   D_1(z) &\Defeq X \cdot z\mbox{\@,}
   \end{align*}  
for multiplying a stream by a constant, adding two streams, copying, 
and delaying a stream. 
Such a stream circuit is visualized
in Figure~\ref{fig:stream-circuit-1}\@. 
\end{example}

 \begin{figure}[t]
    \centering
    \begin{center}
    \tikzstyle{int}=[draw, fill=blue!20, minimum size=3em]
    \tikzstyle{comp}=[ minimum size=1em]
    \tikzstyle{init} = [pin edge={to-,thin,black}]
    \begin{tikzpicture}[node distance=2.5cm,auto,>=latex']
    \node [coordinate] (input) {input};
    \node [int] (a) [right of=input, node distance=2cm] {$A$};
    \node [comp] (h1) [right of=a, node distance=2cm] {$\stackrel{h_3}{\circ}$};
    \node [int] (c) [right of=h1, node distance=2cm] {$C$};
    \node [comp] (h2) [above of=a, node distance=2cm] {$\stackrel{h_1}{\circ}$};
    \node [int] (d1) [right of=h2, node distance=2cm] {$D_1$};    
    \node [comp] (h3) [right of=d1, node distance=2cm] {$\stackrel{h_2}{\circ}$};
    \node [coordinate] (output) [right of=c, node distance=2cm]{};
    \path[->] (input) edge node {$z$} (a);
    \path[->] (a) edge node  {} (h1);
    \path[->] (h1) edge node  {} (c);
    \path[->] (c) edge node {$y$} (output);
     \path[->] (c) edge node  {} (h3);
     \path[->] (h3) edge node  {} (d1);
     \path[->] (d1) edge node  {} (h2);
     \path[->] (h2) edge node  {} (a);
    \end{tikzpicture}
    \end{center}
    \caption{Finite Stream Circuit. \label{fig:stream-circuit-1}}
 \end{figure}
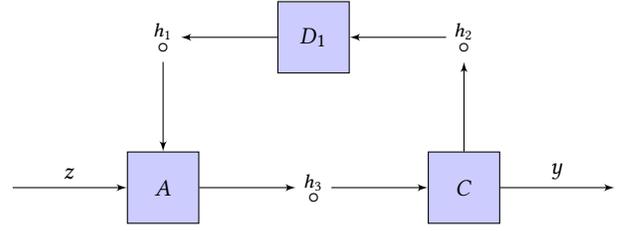

\begin{example}
We can systematically analyze stream circuits, such as the one in Figure~\ref{fig:stream-circuit-1}, by solving the underlying system of stream equations 
    $$
        h_1 = X \cdot h_2~~~
        h_3 = z + h_1~~~
        h_2 = h_3~~~
        y = h_3\mbox{\@,}
    $$
for the output stream $y$ to obtain
   $
      y = (\nicefrac{1}{1 - X}) \cdot z 
  $\@.
Thus, the (rational) stream transformer $\Stream{\Real} \to \Stream{\Real}$ that is 
implemented by the feedback circuit in Figure~\ref{fig:stream-circuit-1} is given by
    $
    z \mapsto (\nicefrac{1}{1 - X}) \cdot z
    $\@.
Since $\nicefrac{1}{1 - X} = (1, 1, 1, \ldots)$, the output stream
$y$ is, by stream multiplication~(\ref{multiply}),
of the form $(\sum_{k=0}^n [X^k]z)_{n \in \Nat}$\@.
\end{example}

\paragraph{Composing Transformers.}
System design is based on gradually composing or decomposing stream transformers.
\begin{example}[Functional Composition]\label{ex:stream.transformers-2}
A formal power series  may also 
be viewed as a stream transformer $S(X)$\@, where $X$ 
is now instantiated with an input stream, say, $T(X)$ by the "composition" $S(T(X))$\@. 
Formally,  the {\em functional composition}  $T \Comp S$ 
of two power series 
   $S(X) \Defeq \sum_{k \geq 0} a_k X^k$ 
and 
   $T(X) \Defeq \sum_{k \geq 1} b_k X^k$, that is, $b_0 = 0$\@,
is a power series with coefficients
    \begin{align}\label{def:comp}
        c_k &\Defeq \sum_{n \geq 0}\sum_{j_1 + \ldots + j_n = k} b_n a_{j_1} \ldots a_{j_n}\mbox{\@.}
    \end{align}
%
\end{example}
Inverses for the functional composition (\ref{def:comp}) of two streams
exist under certain conditions, and may be used for {\em reverse} computation. 

\begin{definition}\label{def:basic}
For suitable  
stream transformers $S$, $T$\@:
    \begin{align*}
    \Magic     &\Defeq (\lambda {\_}){\StreamN{A}{n}} \\
    \Abort     &\Defeq (\lambda {\_}){\emptyset} \\
     S ; T      &\Defeq  (\lambda {f})\, T(S(f)) \\
     S \Par T &\Defeq  (\lambda f)\, (T(f), S(f)) \\
    S \Angelic T &\Defeq  (\lambda {f})\,{S(f) \Union T(f)} \\
    S \Demonic T &\Defeq  (\lambda {f})\, {S(f) \Inter T(f)}\mbox{\@.}
    \end{align*}
$S; T$ is {\em sequential composition},
$S \Par T$ is {\em parallel composition},
$S \Angelic T$ realizes {\em angelical} and $S \Demonic T$ {\em demonic} 
nondeterminism, which is possibly unbounded.
\end{definition}
These definitions work for arbitrary vector- and 
multivalued maps $S$, $T$ for which the composition $T(S(f))$ is well-defined
(cmp. Example~\ref{ex:stream.transformers-2})\@.

\paragraph{Refinement.}
A  nondeterministic stream transformer $S$ is a {\em refinement} of another nondeterministic stream transformer $T$, 
written $S \preceq T$ if $S$ is {\em more deterministic} than $T$; that is:
   \begin{align}
   S \preceq T &\Defeq \Forall{f\in \Stream{A}}{S(f) \subseteq T(f)}\mbox{\@,}
   \end{align}
where $\subseteq$ is interpreted component-wise.
Clearly, $\preceq$ is a partial order, and
refinement is compatible with composition operators.
\begin{lemma}[Compatibility]
Let $S_1 \preceq T_1$, $S_2 \preceq T_2$; then:
   \begin{enumerate}
       \item $S_1; S_2 \preceq T_1; T_2$
       \item $S_1 \Par S_2 \preceq T_1 \Par T_2$
       \item $S_1 \Angelic S_2 \preceq T_1 \Angelic T_2$
       \item $S_1 \Demonic S_2 \preceq T_1 \Demonic T_2$
    \end{enumerate}
\end{lemma}



\section{Causality}\label{sec:causality}
We restrict our investigation to stream transformers whose outputs are determined by their history of inputs.
The {\em prefix} $\Prefix{f}{n}$ of length $n \in \Nat$ for a stream $f$ is the finite word $([X^0]f) \ldots ([X^{n-1}]f)$ of the first $n$ coefficients of $f$\@. 
Now, for a set $F$ of streams,  $\Prefix{F}{n}$  denotes 
the set of prefixes 
$\Prefix{f}{n}$ for $f \in F$\@. 
Similarly, $[X^k]F$ denotes the set $\Setc{[X^k]f}{f \in F}$\@.
Causal stream transformers, at least in the single-valued case, are discussed in~\cite{broy2012specification}\@. 

\begin{definition}[Causal Stream Transformers]\label{def:causal.stream-transformer}
The nondeterministic stream transformer $T: \StreamN{A}{n} \to \Powerset{\StreamN{B}{m}}$ 
is {\em $\delta$-causal}, for $\delta \in \Nat$,
if for all $k \in \Nat$ and $f, g \in \StreamN{A}{n}$
   \begin{align}
       \Prefix{f}{k} = \Prefix{g}{k}
          &\text{~implies~} [X^{k + \delta}]T(f) = [X^{k + \delta}]T(g)\mbox{\@.}
   \end{align}
In case $T$ is $0$-causal then $T$ is also said to be {\em (weakly) causal}, and if $T$ is $\delta$-causal for $\delta > 0$ then $T$ is {\em strongly causal}\@. 
\end{definition}
Every $\delta$-causal transformer is also $\delta'$-causal for $\delta' < \delta$\@.
In particular, every strongly causal transformer is also  weakly causal.

An alternative characterization of $\delta$-causality to the one in Definition~\ref{def:causal.stream-transformer} is easily established by natural induction on $k$\@. 
\begin{proposition}
A stream transformer $T$ is $\delta$-causal, for $\delta \geq 0$, if and only if
for all $f, g \in \Stream{A}$ and $k \in \Nat$,
   \begin{align}
    \Prefix{f}{k} = \Prefix{g}{k} &\text{~~implies~~} \Prefix{T(f)}{k + \delta} = \Prefix{T(g)}{k + \delta}\mbox{\@.} 
   \end{align}
\end{proposition}

Strongly causal stream transformer are often used to specify clocked, hardware-like finite structures. 
However, when designing a system, it is often preferable to work with the larger class of weakly causal stream transformers~\cite{benveniste1991synchronous}\@. 

\begin{example}\label{ex:register}
A clocked $N$-bit register holding a value belonging to the alphabet 
$A \Defeq \Bool^{N}$ induces  a causal stream transformer
$\Stream{(A \Cross \Bool)} \to \Stream{A}$, 
where the second input $\Bool \Def \{0,1\}$ corresponds 
to an enabling stream for stating whether the incoming value 
on the first component ought to be loaded into the register
or be ignored at a given clock tick\@.
\end{example}

\begin{example}\label{ex:rational.stream.transformers}
The basic stream transformers in Example~\ref{ex:stream-circuits} for constructing stream circuits
are all ($0$-) causal, and $D_1$ is also $1$-causal.
\end{example}

\begin{example}[Consing]
Consider, for $a \in A$, the family 
$Cons_a: \Stream{A} \to \Stream{A}$ of stream transformers, 
which map $z \mapsto a + X \cdot z$\@.
These are $1$-causal transformers. 
$Cons_0$, in particular, yields the unit delay $D_1$ from Example~\ref{ex:rational.stream.transformers}\@. 
\end{example}

\begin{example}
The stream transformer 
which returns the stream $1$ if all elements of its argument stream are positive, and
the stream $0$ otherwise, is not causal.  
\end{example}

%


\paragraph{Mealy Realizability.}
A Mealy machine is an intensional description of a causal stream transformer $T$,
where the state holds enough information to  determine $[X^k]T(f)$ from $[X^k]f$\@.
More precisely, the set of weakly causal stream transformers $\Stream{A} \to \Stream{B}$ is isomorphic to the functions $A^+ \to B$ from finite words over $A$ in $A^+$ to $B$, and a Mealy machine, possibly containing an infinite number of states, can be constructed for realizing such a word level function. 

\paragraph{Composition.} Causality is compositional for the basic
operators on transformer from Definition~\ref{def:basic}\@.
\begin{proposition}\label{prop:causal-accumulates}
Let $S$ be $\delta_S$-causal and $T$ 
be $\delta_T$-causal; then: 
  \begin{enumerate}
      \item $\Magic$, $\Abort$ are weakly causal;
       \item $S ; T$ is $(\delta_S + \delta_T)$-causal;
      \item $S \Par T$ is $\Min{\delta_S}{\delta_T}$-causal;
      \item $S \Angelic T$ and $S \Demonic T$ are $\Max{\delta_S}{\delta_T}$-causal. 
  \end{enumerate}
\end{proposition}
As a consequence, $\delta$-causal stream transformers, for $\delta > 0$, 
are not preserved under composition, which results in 
an unfortunate inflexibility in system design.
For instance, if we want to represent a stream transformer $T$
as a sequential composition $T_1; T_2$ of two $1$-causal stream transformers then we always have to accept a delay by at least two.
Both weakly and strongly causal stream transformers, however, 
are preserved, as a consequence of Proposition~\ref{prop:causal-accumulates} under sequential and parallel composition. 
\begin{lemma}
For weakly (strongly) causal stream transformers $S$, $T$ (with suitable domains and codomains)
$S; T$, $S \Par T$, $S \Demonic T$, and $S \Angelic T$  are weakly (strongly) causal\@. 
\end{lemma}

\paragraph{Fixpoints.}

Strongly causal stream transformers have a unique fixpoint. 
This can be shown, for example, by constructing 
recursive Mealy machines whose transitions may also depend on outputs~\cite{pradic2020some}\@.
In contrast, fixpoints for weakly causal transformers may not exist and they may not be unique.
\begin{example}~\label{ex:causal-nonfix}
\begin{enumerate}
    \item $Succ(z) \Def z + 1$ has no fixpoint in $\Stream{\Real}$;
    \item Every stream in $\Stream{A}$ is a fixpoint of the identity transformer $Id(z) \Def z$\@. 
\end{enumerate}
\end{example}
Synchronous languages such as Lustre~\cite{halbwachs2012synchronous} or Esterel~\cite{berry1992esterel} therefore 
use syntactic restrictions for avoiding causal loops.
In contrast, the approach taken in the {\em symbolic analysis laboratory} SAL~\cite{bensalem2000overview} 
relies on {\em verification conditions} asserting the absence of causal loops.
\begin{example}\label{ex:sal}
Let $c$, $d$, $e$ be Boolean-valued streams.
Then a solution $(f, g) \in \Stream{(\Bool \Cross \Bool)}$ of
the stream equalities\footnote{
$\ITE{c}{f}{g}$ denotes the (pointwise) conditional stream "if $c$ then $f$ else $g$"\@.
}
  \begin{align*}
   f &= \ITE{c}{\lnot{g}}{d} \\
   g &= \ITE{c}{e}{f}\mbox{\@}
  \end{align*}
is a fixpoint of 
   \begin{align*}
   T(f, g) &\Def \ITE{c}{(\lnot{g}, e)}{(d, f)}\mbox{\@.}
   \end{align*}
There is no causal loop, because $f$ is causally dependent on $g$ only when $c$ is true, and vice versa only when it is false. 
Therefore, these stream equations are acceptable in SAL.
\end{example}
The approach of SAL of semantically characterizing causality errors is more general than the synctactic restrictions, say, in Lustre.
But it is also undecidable in general, since it can depend on arbitrary data properties.


\section{Contraction}\label{sec:contractions}

We develop metric-based characterizations both of weakly and of strongly causal stream transformers.

\subsection{Multivalued Contraction}

\begin{definition}[Contractions for Multivalued Maps]\label{def:contractions}
Let $(M, d_M)$, $(M, d_N)$ be metric spaces\@. 
A multivalued map $T: M \to \Powerset{N}$ is said to be an 
{\em $l$-contraction}, for $0 \leq l \leq 1$, if
for all $x, y \in M$ and for all $u \in T(x)$ 
there exists $v \in T(g)$ such that 
   \begin{align}\label{contraction-condition}
   d_M(u, v) &\leq l \cdot d_N(x, y)\mbox{\@.} 
   \end{align}
If $l < 1$ then the multivalued $T$ is said to be {\em (strictly) contractive}, and if $l = 1$ then $T$ is {\em nonexpansive}\@. 
\end{definition}
Contraction for multivalued maps $x \mapsto \{x\}$  coincides with the usual notion of contraction on singlevalued 
maps as stated, for example, in Section~\ref{sec:preliminaries}\@. 

\begin{lemma}\label{lem:causal-char-power}
A stream transformer $T: \StreamN{A}{n} \to \Powerset{\StreamN{A}{m}} \setminus \emptyset $
is $\delta$-causal, for $\delta \geq 0$, if and only if it is $2^{-\delta}$-contractive. 
\end{lemma}
\begin{proof}
Let $f, g \in \Stream{A}$  and let $k \in \Natinf$ such that $d(f, g) = 2^{-k}$, 
and therefore $\Prefix{f}{k} = \Prefix{g}{k}$ ($\Prefix{h}{\infty} \Defeq h$)\@.
\begin{enumerate}
\item ($\Rightarrow$) 
     By assumption, $T$ is $\delta$-causal, and therefore
            $$\Prefix{T(f)}{k + \delta}  = \Prefix{T(g)}{k + \delta}\mbox{\@.}$$ 
        Consequently, for all $u \in T(f)$
        there exists $v \in T(g)$ with
             $\Prefix{u}{k + \delta}  = \Prefix{v}{k + \delta}$,
        that is $d(u, v) \leq 2^{-(k+\delta)} =  2^{-\delta} \cdot d(f, g)$\@. 
\item ($\Leftarrow$)
     By assumption, for all $u \in T(f)$ there exists $v \in T(g)$ such that
            $$d(u, v) \leq 2^{-\delta} \cdot d(f, g) = 2^{-(k + \delta)}\mbox{\@,}$$
     and therefore $\Prefix{T(f)}{k + \delta} \subseteq \Prefix{T(g)}{k + \delta}$\@. 
     Similarly, also $\Prefix{T(g)}{k + \delta} \subseteq \Prefix{T(f)}{k + \delta}$ holds.  
\end{enumerate} 
This finishes the proof.
\end{proof}
From Lemma~\ref{lem:causal-char-power} one obtains metric-based 
characterizations for weak and also for strong causality.
\begin{corollary}[Causality-Contraction Connection]\label{thm:causal-equiv-contractive}
A nondeterministic stream transformer $T: \StreamN{A}{n} \to \Powerset{\StreamN{B}{m}} \setminus \emptyset$ is:
   \begin{enumerate}
       \item Weakly causal if and only if it is nonexpansive.
       \item Strongly causal if and only if it is contractive.
   \end{enumerate}
\end{corollary}
Note that this correspondence can be extended for shrinking transformers together with a new notion 
of causality with non-uniform bounds on determined outputs.

\subsection{Fixpoints of Stream Transformers}

Recall that a {\em strictly contracting} or {\em shrinking} map $T$ 
satisfies the strict inequality $d(T(f), T(g)) < d(f, g)$ for all $f, g$\@.
Now, any shrinking map in a spherically complete ultrametric space has a unique fixpoint  (\cite{petalas1993fixed}, Theorem~1)\@.
\begin{definition}\label{def:loop}
For $T: \StreamN{A}{n} \to \StreamN{A}{n}$
a strictly contractive stream transformer, 
$\LOOP{T}$ is the unique fixpoint of
${\mathcal T} \Defeq (\lambda S)\, S; T$\@. 
\end{definition}
$\LOOP{T}$ is well-defined by (\cite{petalas1993fixed}, Theorem~1), since, using Lemma~\ref{lem:product.ultrametric.complete}, 
maps of (deterministic) stream transformers form an ultrametric and spherically complete space, and ${\mathcal T}$ is strictly contractive in this space.
   \begin{align*}
    d({\mathcal T}(S_1), {\mathcal T}(S_2)) 
    &=  d((S_1; T), (S_2; T)) \\
    &= \sup_{f} d(T(S_1(f)), T(S_2(f))) \\
    &< \sup_{f} d(S_1(f), S_2(f)) \\
    &= d(S_1, S_2)\mbox{\@.}
   \end{align*}
Using fixpoint results for strictly contractive
multivalued maps (Theorem~2 in~\cite{kubiaczyk1996multivalued})
Definition~\ref{def:loop} is generalized to obtain
a unique fixpoint of transformers of strictly contracting multivalued stream 
transformers.

\subsection{Lipschitz Contraction}

We show that contraction on multivalued maps
with nonempty compact codomains
coincides with contraction with respect to
the Hausdorff metric as defined in~\cite{nadler1969multi}\@. 
 
\begin{definition}[Lipschitz Contraction]\label{def:lipschitz.contractions}
Let $(M, d_M)$ and $(N, d_N)$ be metric spaces\@. 
A multivalued map $T: M \to \CB{N}$
is a {\em Lipschitz mapping} if and only if
       \begin{align}\label{hausdorff-contraction}
           \Hd{d_N}{T(x)}{T(y)} \leq l \cdot d_M(x, y)\mbox{\@,}
       \end{align}  
for all $x, y \in M$, 
where $l > 0$ is the {\em Lipschitz} constant for $T$\@. 
In these cases we also say that $T$ is $l$-{\em Lipschitz}\@.
Furthermore, if $l = 1$ then $T$ is {\em Lipschitz nonexpansive},
and if $0 < l < 1$ then $T$ is {\em Lipschitz contracting}\@.  
\end{definition}

A multivalued mapping $T$ with Lipschitz constant $l$ is uniformly continuous, 
since for arbitrary $\varepsilon > 0$ set $\delta \Defeq \nicefrac{\varepsilon}{l}$ to obtain 
$$d(T(x),T(y)) \leq l \cdot d(x,y) < l \cdot \delta = \varepsilon$$
from $d(x, y) < \delta$\@.
    The continuous image of a compact set is compact.
Assume an open cover of $T(K)$\@.
As $T$ is continuous, the inverse image of those open sets form an open cover for $K$\@.
Since $K$ is compact, there is a finite subcover of $T(K)$, and, by construction, the images of the finite subcover give a finite subcover of $T(K)$, 
and therefore $T(K)$ is compact.
\begin{lemma}\label{lem:comp-closed-under-contracting-maps}
If $T: M \to \Compact{N}$ is continuous and $K \in \Compact{M}$ 
then $T(K) \in \Compact{N}$\@.
\end{lemma}
In particular, the image $T(K)$ of a compact set $K$ with respect to a Lipschitz map $T$ is compact. 

We present some elementary results which will be used in later 
sections~(\cite{nadler1969multi})\@.
The proofs of many of these facts are straightforward (see also Proposition~\ref{prop:causal-accumulates})\@. 
\begin{proposition}
 If $T: L \to \Compact{M}$ is $l_T$-Lipschitz and $S: M \to \Compact{N}$ is $l_S$-Lipschitz,
then $S \circ T: L \to \Compact{N}$ is $(l_T \cdot L_S)$-Lipschitz\@. 
\end{proposition}
\begin{proposition}\label{prop:angelic.lipschitz.closed}
If $S, T: M \to \Compact{N}$ are $l_S$- and $l_T$ -Lipschitz , respectively,
then  $(S \Angelic T): M \to \Compact{N}$ is $\Max{l_S}{l_T}$-Lipschitz\@. 
\end{proposition}
Proposition~\ref{prop:angelic.lipschitz.closed} gives a technique for constructing a Lipschitz mapping from 
a finite number of single-valued Lipschitz mapping by "unioning their graphs at each point"\@. 
The closure condition (\ref{prop:angelic.lipschitz.closed}) of Lipschitz maps under {\em angelic nondeterminism} $S \Angelic T$ can be generalized to an arbitrary  ${\mathcal I}$-index family $(T_\iota)_{\iota \in {\mathcal I}}$~\cite{nadler1969multi}\@. 


For stream transformers with nonempty compact images,
Lipschitz contraction (Definition~\ref{def:lipschitz.contractions}) is 
equivalent to the notion of contraction mappings in Definition~\ref{def:contractions}\@. 

\begin{lemma}\label{lem:hoare-equiv-hausdorff}
 Let $(M, d_M)$, $(N, d_N)$ be metric spaces, 
  $T: M \to \Compact{N}$ a multivalued map, 
  and constant $l$ with $0 < l \leq 1$; then: 
    $T$ is an  $l$-contraction if and only if $T$ is $l$-Lipschitz.
\end{lemma}
\begin{proof}   
For given $x, y \in M$
we may, without loss of generality,  assume
       $$\Hausdorff{d_N}(T(x), T(y)) = \sup_{u \in T(x)}(d(u, T(y))\mbox{\@.}$$
\begin{enumerate}
\item ($\Rightarrow$)
As a consequence of this assumption,
           $
             \Hausdorff{d_N}(T(x), T(y)) \leq d(u, T(y))
           $
for all $u \in T(x)$\@.
Since, $T$ is $l$-contractive, there is, by definition, $v_0 \in T(y)$ such that
    $d_N(u, v_0) \leq l\, d(x, y)$, and consequently:
   $$
    d(u, T(y)) = \inf_{v \in T(y)} d_N(u, v)
               \leq d_N(u, v_0)
               \leq l \cdot d(x, y)\mbox{\@.}
   $$
Altogether, 
   $
   \Hausdorff{d_N}(T(x), T(y)) = \sup_{u \in T(x)} d(u, T(y)) \leq l \cdot d(x, y)
   $\@.
\item ($\Leftarrow$)
   For all $u \in F$, 
   $$
   \inf_{v \in T(y)} d(u, v) = d(u, T(y))
                             \leq \Hausdorff{d_N}(T(x), T(y))\mbox{\@.}
   $$
By compactness of $T(y)$, 
there exists $v_0 \in T(y)$
such that 
    $$
    d(u, v_0) = d(u, T(y)) \leq \Hausdorff{d_N}(T(x), T(y))\mbox{\@.} 
    $$
\end{enumerate}
\end{proof}

The proof of the left-to-right statement of Lemma~\ref{lem:hoare-equiv-hausdorff} 
only requires the nonempty images of $T$ to be closed (and bounded), and not necessarily compact\@. 
\noindent
Lemma~\ref{lem:hoare-equiv-hausdorff} and Corollary~\ref{thm:causal-equiv-contractive} together
yield the correspondence between causal and  Lipschitz contractive maps. 
\begin{corollary}[Causality-Contraction Connection (Lipschitz)]\label{thm:causal-equiv-hausdorff}
The nondeterministic stream transformer $T: \StreamN{A}{n} \to \Compact{\StreamN{B}{m}}$ is:
   \begin{enumerate}
       \item Weakly causal if and only if it is Lipschitz nonexpansive.
       \item Strongly causal if and only if it is Lipschitz contractive.
   \end{enumerate}
\end{corollary}

\section{Fixpoints}\label{sec:fixpoints}

Fixpoints for strictly causal functions are constructed, for example, on the basis of recursive Mealy machines~\cite{pradic2020some}\@.
We pursue, however, a less syntax- and machine-oriented path, for 
solving stream equations, and we compute fixpoints for vector- and multivalued maps. 

We are interested in fixpoints of
stream transformers $T: \StreamN{A}{n} \to \Powerset{\StreamN{A}{n}}$, 
that is, vectors of streams $f^*$ with
   \begin{align}\label{nondeterminist.fixpoint}
    f^* &\in T(f^*)\mbox{\@.}
   \end{align}
For deterministic maps $T$ this inequality reduces to the fixpoint equality $f^* = T(f^*)$\@. 
Since the mapping $\iota: \StreamN{A}{n} \to \Compact{\StreamN{A}{n}}$, given by $f \mapsto \{f\}$ 
for each $f \in \StreamN{A}{n}$, is an isometry, the fixpoint theorems in this paper for multivalued mappings are generalizations of their single-valued analogues.

\subsection{Set Transformers}

We show that the {\em weakest pre} and {\em strongest post} of Lipschitz 
contraction (equivalently, strongly causal) maps have a unique fixpoint, and we
formulate a corresponding fixpoint induction principle.
Moreover, the unique fixpoint of $\SP{T}$ includes all
fixpoints of the multivalued map $T$\@. 

\begin{definition}
Let $T: M \to \Compact{N}$ be a Lipschitz map; then: 
   \begin{enumerate}
       \item The {\em strongest post} $\SP{T}: \Compact{M} \to \Compact{N}$ is the set transformer 
                 \begin{align} 
                   \SP{T}(P) &\Defeq T(P)\mbox{\@;}
                 \end{align}
       \item The {\em weakest pre} $\WP{T}: \Compact{M} \to \Compact{N}$ is the set transformer
              \begin{align}
            \WP{T}(Q) &\Defeq \Setc{x \in \Compact{M}}{T(x) \subseteq Q}\mbox{\@.} 
              \end{align}
   \end{enumerate}
\end{definition}
It follows that the $\SP{T}$ and $\WP{T}$ are {\em adjoint} in that
    \begin{align}
    \SP{T}(P) \subseteq Q &\Iff  P \subseteq \WP{T}(Q)
    \end{align}
for all  $P \in \Compact{M}$, $Q \in \Compact{N}$\@.
As a consequence, $P \subseteq (\WP{T} \Comp \SP{T})(P)$,
$(\SP{T} \Comp \WP{T})(Q) \subseteq Q$, and 
$\SP{T}$ and $\WP{T}$ are both {\em monotonic} with respect to set inclusion.

\begin{example}[WP Stream Calculus]\label{wp.calculus}~\\
Let $Q \in \Compact{\StreamN{A}{n}}$ and 
$S, T: \StreamN{A}{n} \to \Compact{\StreamN{B}{m}}$ 
strongly causal; then: 
   \begin{align*}
         \WP{\Magic}(Q) &= \StreamN{A}{n} \\
         \WP{\Abort}(Q) &= \emptyset \\
         \WP{S; T}(Q) &= \WP{S}(\WP{T}(Q)) \\
    \WP{S \Demonic T}(Q) &= \WP{S}(Q) \Inter \WP{T}(Q) \\
    \WP{S \Angelic T}(Q) &= \WP{S}(Q) \Union \WP{T}(Q) \\
    \WP{T^+}(Q) &= \WP{T^+}(\WP{T}(Q))
   \end{align*}
\end{example}

\begin{example}
A {\em Hoare-like stream contract} $\{P\}\,T\,\{Q\}$ for a nondeterministic stream transformer $T$ with
{\em precondition} $P$ and {\em postcondition} $Q$ holds if and only if $P \subseteq \WP{T}(Q)$, or, equivalently, $\SP{T}(P) \subseteq Q$\@. 
Corresponding rules for a Hoare-like calculus are derived from the weakest precondition calculus (Lemma~\ref{wp.calculus})\@. 
\end{example}

\begin{proposition}\label{prop:sp.wp.lipschitz}
If $T: M \to \Compact{N}$ is $l$-Lipschitz
then both $\SP{T}$ and $\WP{T}$ are $l$-Lipschitz.
\end{proposition}
Consequently, 
if $T$ is Lipschitz contracting, then both $\SP{T}$ and $\WP{T}$ are
Lipschitz contracting, 
and we obtain unique fixpoints for $\SP{T}, \WP{T}: \Compact{M} \to \Compact{M}$ from  Banach's contraction principle, 
since $(\Compact{M}, \Hausdorff{d_M})$ is a Cauchy complete metric space.

\begin{lemma}\label{lem:fixpoints.wp.sp}
Let $(M, d_M)$ be a metric space. 
If $T: M \to \Compact{M}$ is Lipschitz contractive then 
\begin{enumerate}
\item  $\SP{T}, \WP{T}: \Compact{M} \to \Compact{M}$
have unique fixpoints, say, $fix(\SP{T})$ and $fix(\WP{T})$, respectively\@.
\item  For $F_{k+1} = \SP{T}(F_{k})$, 
$G_{k+1} = \WP{T}(G_{k})$
and arbitrary $F_0, G_0 \in \Compact{M}$\@,
   \begin{align*}
    fix(\SP{T}) &= \lim_{k\to\infty} F_k \\
    fix(\WP{T}) &= \lim_{k\to\infty} G_k\mbox{\@.}
   \end{align*}
\item
For $k \in \Nat$\@:
   \begin{align*}
    \Hausdorff{d_M}(F_k, fix(\SP{T})) &\leq \nicefrac{l^k}{(1 - l)} \cdot  \Hausdorff{d_M}(F_0, F_1) \\
    \Hausdorff{d_M}(G_k, fix(\WP{T})) &\leq \nicefrac{l^k}{(1 - l)} \cdot  \Hausdorff{d_M}(G_0, G_1)\mbox{\@.}
   \end{align*}
\end{enumerate}
\end{lemma}

Compared with the usual Knaster-Tarski fixpoint iteration of Scott-continuous transformers, the main advantage of the iteration in Lemma~\ref{lem:fixpoints.wp.sp} is that, at any iteration, there is a quantitative measure of the distance to the fixpoint and also of the progress towards this fixpoint. 
On the other hand, the iterations in Lemma~\ref{lem:fixpoints.wp.sp}
generally neither are under- nor overapproximations of fixpoints.

\begin{lemma}\label{lem:fixpoint.approximation}
Every fixpoint of a contractive $T: M \to \Compact{M}$
is included in $\fix{\SP{T}}$\@.
\end{lemma}
\begin{proof}
For a fixpoint $f^*$ of $T$, that is, $f^* \in T(f^*)$ define the iteration $F_0 \Defeq \{f^*\}$
$F_{k+1} \Defeq \SP{T}(F_k) = \Setc{g \in T(f)}{f \in F_k}$\@. 
By induction on $k$, $f^* \in F_k$ for each $k\in \Nat$, and therefore $f^* \in \lim_{k\to\infty} F_k = \fix{\SP{T}}$\@. 
\end{proof}
The subset relation in Lemma~\ref{lem:fixpoint.approximation} may be proper~\cite{nadler1969multi}\@.
When applying Lemma~\ref{lem:fixpoints.wp.sp} to streams $\StreamN{A}{n}$ and its ultrametric $d$, then the convergence bounds can be improved to $l^k$ with $l = (\nicefrac{1}{2})^\delta$ for some $\delta > 0$, since $d$ is bounded and the application of the strengthened triangle inequality in the proof of Banach's contraction principle yields the improved bound $l^k$\@. 
\begin{theorem}\label{thm:main0}
Let $T: \StreamN{A}{n} \to \Compact{\StreamN{B}{m}}$ be
$\delta$-causal, for $\delta > 0$\@. 
With the notation and the iterations $F_k$, $G_k$ as 
in Lemma~\ref{lem:fixpoints.wp.sp},  the 
unique fixpoints $\fix{\SP{T}}$ and $\fix{\WP{T}}$
are obtained as the limits of $F_k$ and $G_k$, respectively,
as $k\to\infty$\@.
Furthermore, for $k \in \Nat$\@:
   \begin{align*}
    \Hausdorff{d}(F_k, \fix{\SP{T})} &\leq (\nicefrac{1}{2})^{k\delta} \\
    \Hausdorff{d}(G_k, \fix{\WP{T})} &\leq (\nicefrac{1}{2})^{k\delta}\mbox{\@.}
   \end{align*}
\end{theorem}


\begin{corollary}[Induction Principle for $\SP{}$]\label{cor:induction.sp}~\\
Let $P \subseteq \StreamN{A}{n}$ be a closed set of streams
and $T: \StreamN{A}{n} \to  \Compact{\StreamN{A}{n}}$ contracting; then
   \begin{align}
   \Forall{Q \subseteq P}{\SP{T}(Q) \subseteq P}
      &\Imp \fix{sp_T} \subseteq P\mbox{\@.}
   \end{align} 
\end{corollary}
\begin{proof}
For $F_0 \Defeq \emptyset$ and $F_{k+1} = \SP{T}(F_k)$ we get, by induction on $k$ and the assumption above, that
$F_k \subseteq P$ for all $k \in \Nat$\@. 
The closed set $P$ contains all of its limit points, and therefore  $\fix{\SP{T}} = \lim_{k\to\infty} F_k \subseteq P$\@. 
\end{proof}
For the adjointness of $\SP{T}$ and $\WP{T}$ we
might also use the equivalent assumption
  $\Forall{Q \subseteq P}{Q \subseteq \WP{T}(P)}$
in the induction principle for $\SP{T}$ (see Corollary~\ref{cor:induction.sp})\@.

\subsection{Contractive Maps}

The constant map $T_\emptyset$ to the emptyset  can not have fixpoints, since $f \notin \emptyset = T_\emptyset(f)$\@.  
Also, we restrict the codomain of multivalued mappings to closed sets only, since any closed subset of a Cauchy complete space is Cauchy complete.

\begin{lemma}\label{lem:fixpoint.power}
A  nondeterministic contractive stream transformer $T$ 
has a fixpoint if (1) $T$ is not 
the constant map to the empty set, and
(2) $T(f)$ is closed for all streams $f$\@. 
\end{lemma}
\begin{proof} 
Since $T$ is not the constant map to the empty set, there exist $f_0, f_1 \in \Stream{A}$
such that $f_1 \in T(f_0)$\@.
Now, since $T$ is contractive
there is $0 \leq l < 1$ 
and a $f_2 \in T(f_1)$ with
$d(f_1, f_2) \leq l \cdot d(f_0, f_1)$\@.
In this way, we recursively construct a sequence $(f_k)$ such that $f_{k+1} \in T(f_k)$ and
  $$
   d(f_{k+1}, f_{k+2}) \leq  l \cdot d(f_k, f_{k+1}) \leq \cdots \leq l^k \cdot d(f_0, f_1) \leq l^k\mbox{\@,}  
  $$
for every $k \in \Nat$\@. 
The strengthened triangle inequality for the ultrametric $d$ yields
  \begin{align}~\label{cauchy.ineq}
  & d(f_k, f_m)  \leq \mathrm{max}(d(f_k, f_{k+1}), \ldots, d(f_{m-1}, f_{m}))
              \leq l^k\mbox{\@,}
  \end{align}
which implies that the sequence $(f_k)$ is Cauchy in $(\StreamN{A}{n}, d)$\@.
For Cauchy completeness, $(f_k)$ therefore converges 
to some $f^* \in \StreamN{A}{n}$\@. 
From the inequality~(\ref{cauchy.ineq}) we conclude 
in the limit $m \to \infty$ that
      \begin{align}\label{ineq:seq}
          d(f_k, f^*) &\leq l^k\mbox{\@.}
      \end{align}
Since $T$ is a contraction, there is a sequence $(g_k)_{k \in \Nat}$
such that the inequality
          $d(f_{k+1}, g_k) \leq l\cdot d(f_k, f^*)$ holds.
Therefore, by inequality~(\ref{ineq:seq})\@:
      \begin{align*}
          d(f^*, g_k) &\leq \Max{d(f^*, f_{k+1})}
                                            {d(f_{k+1}, g_k)}  \\
                      &\leq \Max{d(f^*, f_{k+1})}{l \cdot d(f_{k}, f^*)} \\
                      &= l^{k+1}\mbox{\@.}
      \end{align*}
Now, $\lim_{k \to\infty} g_k = f^*$\@,
since  $\lim_{k \to\infty} d(f^*, g_k) = 0$\@.
But $g_k \in T(f^*)$ and all images of $T$ are closed, 
and therefore  the limit $f^*$ of the sequence $(g_k)$ 
is an element of $T(f^*)$\@.
Altogether, $f^* \in T(f^*)$\@. 
\end{proof}
In general, fixpoints of contractive multivalued maps in the sense of Definition~\ref{def:contractions} are not unique, but a "slightly" stronger requirement on contractiveness implies 
uniqueness of fixpoints.
\begin{definition}[Strong Contractions]\label{def:strong-contractions}~\\
Let $T$ be a nondeterministic stream transformer and $0 \leq l < 1$\@. 
If for all $f, g \in \StreamN{A}{n}$ and for all $u \in T(f)$ and $v \in T(g)$ such that 
   $
   d(u, v) \leq l \cdot d(T(f), T(g))$\@,
then $T$ is a {\em strong contraction} 
with Lipschitz constant $l$\@. 
\end{definition}
If $\emptyset \notin T(\StreamN{A}{n})$ then
every strong $l$-contraction also is a (weak) $l$-contraction 
in the sense of Definition~\ref{def:contractions}\@.
Moreover, a deterministic stream transformer $T$ is $l$-contractive if and only if it is strongly $l$-contractive. 
\begin{lemma}\label{lem:fixpoint.power.unique}
A nondeterministic strongly contractive 
stream transformer has at most one fixpoint.
\end{lemma}
\begin{proof}
We assume fixpoints $f^*$, $g^*$ of $T$\@.
Since $T$ is strongly contractive, there is an $l$ with $0 < l < 1$ such that
   $$
   d(f^*, g^*) \leq \max\Setc{d(u, v)}{u \in T(f^*), v \in T(g^*)} \leq l \cdot d(f^*, g^*)\mbox{\@.}
   $$
Thus, $d(f^*, g^*) = 0$, and
any two fixpoints of $T$ are equal. 
\end{proof}
For deterministic maps, the fixpoint iteration with $f_{k+1} \in T(f_k)$,
as constructed in the proof of Lemma~\ref{lem:fixpoint.power} reduces to the
{\em Picard iteration} $f_{k+1} = T(f_k)$\@.
Moreover, uniqueness of fixpoints for deterministic, $\delta>0$-causal maps 
can  also be shown directly, since  the equality of two arbitrary fixpoints 
$f^*$, $g^*$ directly follows from 
$d(f^*, g^*) =  d(T(f^*), T(g^*)) \leq 2^{-\delta} \cdot d(f^*, g^*)$\@. 
The following result directly follows from Lemmata~\ref{lem:causal-char-power},~\ref{lem:fixpoint.power}, and~\ref{lem:fixpoint.power.unique}\@.
\begin{theorem}\label{thm:main1}
A strongly causal stream transformer $T$ has a fixpoint if it is
(1) not the constant map to the empty set, and (2) $T(f)$ is closed for each stream $f$\@. 
If, in addition, $T$ is strongly contractive then this fixpoint is unique.
\end{theorem}
As an immediate consequence of Theorem~\ref{thm:main1} we obtain as a special case a reformulation of Banachs's fixpoint principle for streams.
\begin{corollary}\label{prop:fixpoints.deterministic}
Strongly causal deterministic stream transformers
have a unique fixpoint. 
\end{corollary}

\paragraph{Fixpoint Induction.}
The fixpoint iteration for strongly causal stream transformers in Theorem~\ref{thm:main1}
is used to derive a {\em fixpoint induction principle}\@.
The overall approach is analogous to deriving fixpoint induction for 
Scott-continuous functions on complete partial orders~\cite{winskel1993formal}\@. 
for strongly causal stream transformers.
\begin{lemma}[Fixpoint Induction]\label{lem:fixpoint-principle-strong}
Let $T$ be a strongly causal stream transformer as in Theorem~\ref{thm:main1} with a unique fixpoint ${\mathrm fix}(T)$
and $P$ a closed set of streams; then\@:
   \begin{align}
      (\Forall{f \in P}{T(f) \in P}) \Imp {\mathrm fix}(T) \in P\mbox{\@.}
   \end{align}
\end{lemma}
\begin{proof}
Let $(f_k)$ be the sequence with $f_{k+1} \in T(f_k)$ for $k \in \Nat$
and  ${\mathrm fix}(T) = \lim_{k\to \infty} f_k$
as constructed in the proof of Lemma~\ref{lem:fixpoint.power}\@.
From the assumption and natural induction on $k$ we obtain that $f_k \in P$ for all $k \in \Nat$\@.
But $P$ is closed, and therefore ${\mathrm fix}(T) \in P$\@.
\end{proof}

\subsection{Shrinking Maps}

%

If $(M, d)$ is a spherically complete ultrametric space, 
then every shrinking map $T: M \to M$ has a unique fixpoint (\cite{petalas1993fixed}, Theorem~1)\@. 
The proof of this result relies on
Zorn's lemma for showing the existence of a maximal, with respect to set inclusion, ball
$B_z$ in the set of balls of the form $B_x \Defeq \BC{x}{d(x, T(x))}$ for $x \in X$; this $z \in X$ is the unique fixpoint of $T$\@. 
A "more constructive proof", not relying on Zorn's lemma,
also shows that there is fixpoint of $T$ 
in every ball of the form $\BC{x}{d(x, T(x))}$ for $x \in X$ (Corollary 5 in~\cite{kirk2012some}, see also~\cite{kirk2014fixed}, Chapter 5.5)\@.

The following statement follows directly from Theorem~2.1 in~\cite{kubiaczyk1996multivalued} and the fact that $\StreamN{A}{n}$ is a spherically complete ultrametric space~\ref{lem:product.ultrametric.complete}\@.
This extension of the results in~\cite{petalas1993fixed} for multivalued 
functions also relies on the application of Zorn's lemma.
Notice that these results 
are stated for complete non-Archimedean normed spaces, but they evidently hold also in  case of spherically complete ultrametric spaces.
\begin{lemma}\label{lem:shrinking}
The nondeterministic stream transformer $T: \StreamN{A}{n} \to \Compact{\StreamN{A}{n}}$ 
has a fixpoint if
 \begin{align}\label{shrinking-condition}
  \Hausdorff{d}(T(f), T(g)) &< d(f, g)
  \end{align} 
for any distinct $f, g \in \StreamN{A}{n}$\@.
\end{lemma} 
Notice that every strongly causal map $T$ as above
satisfies the shrinking condition~(\ref{shrinking-condition})\@. 
\begin{theorem}\label{thm:main2}
Every strongly causal stream transformer $T: \StreamN{A}{n} \to \Compact{\StreamN{A}{n}}$ 
has a fixpoint. 
\end{theorem}
Since every singleton set is compact, we obtain the existence of a fixpoint, in particular, for deterministic maps. 
Uniqueness of this fixpoint follows, for example, from (\cite{kirk2012some}; see also \cite{kirk2014fixed}, Theorem 5.4)\@. 
\begin{corollary}\label{cor:shrinking.fixpoint}
Every strongly causal deterministic stream transformer $T: \StreamN{A}{n} \to\StreamN{A}{n}$ has a unique fixpoint. 
\end{corollary}
The fixpoint for shrinking deterministic stream transformers in Corollary~\ref{cor:shrinking.fixpoint} is obtained as the limit of a transfinite iteration~(\cite{priess2013approximation}; see also~\cite{kirk2014fixed}, Remark 5.5)\@. 
Now, one obtains an induction principle for shrinking (and therefore also strongly causal) deterministic stream transformers
analogously to Lemma~\ref{lem:fixpoint-principle-strong}\@. 


\subsection{Nonexpansive Maps}

In the light of Lemma~\ref{lem:hoare-equiv-hausdorff} we can use causality instead of
the equivalent nonexpansion property ($\Hausdorff{d}(T(f), T(g)) \leq d(f, g)$, for all $f, g$) 
in Theorem 2.2 of~\cite{kubiaczyk1996multivalued}\@. 
\begin{theorem}\label{thm:main3}
If the nondeterministic stream transformer $T: \StreamN{A}{n} \to \Compact{\StreamN{A}{n}}$ is causal
then either $T$ has a fixpoint or there exists a ball $B$ with radius $r > 0$ that that  $d(u, T(u)) = r$ 
for all $u \in B$\@. 
\end{theorem}
\begin{example}
$Succ(x) \Defeq x + 1$ is causal, 
for all  $u \in \Stream{\Real} = \BC{0}{1}$ we have $d(u, Succ(b)) = 1$,  
and  $Succ$ does not have a fixpoint in $\Stream{\Real}$\@. 
\end{example}

Theorem~\ref{thm:main3} provides a concise verification condition to establish
the existence of fixpoints for causal stream transformers in SAL such as the one discussed in Example~\ref{ex:sal}\@.

\section{Related Work}\label{sec:remarks}

The causality-contraction correspondence for the special case of deterministic stream transformers follows from a straightforward unfolding of the definitions. 
This was already observed by 
de~Bakker and Zucker~\cite{de1982processes},
as the basis for using Banach's fixed point theorem to replace the least fixed point approach of denotational semantics based on complete partial orders. 
In particular, de~Bakker's notion of {\em guardedness} is closely related to strong causality~\cite{de2001fixed}\@.
The correspondence between causality and contraction for deterministic stream transformers has also been heavily used in the context of functional reactive programming (for, example~\cite{krishnaswami2011ultrametric,birkedal2012first}) for solving recursive stream equations\@.
Broy~(\cite{broy2023}, A.3) uses this correspondence to solve stream inclusions on timed streams of the form  $x \in F(y)$\@. In his approach, stream inclusions are replaced with stronger equalities of the form $\bar{x} = g(\bar{y})$, where $g$ is a so-called {\em strongly causal representation} of $F$ and $\bar{x}$, $\bar{y}$ are untimed streams corresponding to $x$ and $y$, respectively\@.
Unlike previous work, we extend the
causality-contraction correspondence to 
nondeterministic stream transformers and solve causal stream inclusions directly by corresponding fixed-point results motivated by results from non-Archimedean functional analysis.

The statement and proof of Theorem~\ref{thm:main1} is based on analogous fixpoint results for multivalued maps by Nadler~(\cite{nadler1969multi}, Theorem~5)\@.
Generalizations are described, for example, 
in~\cite{priess2000ultrametric,gajic2001ultrametric,hajimojtahed2016implicit,rao2008common}\@.
In particular, the papers~\cite{priess2000ultrametric,priess1993fixed} are concerned with ultrametric spaces whose distance functions take their values in an arbitrary partially ordered set, not just in the real numbers, and Theorem~\ref{thm:main1} may also have been obtained from the results~(\cite{priess2000ultrametric}, 3.1)\@, but here we prefer to stay in the framework of real-valued metrics.  

Khamsi~\cite{khamsi1993new} uses a generalization of the Banach fixpoint principle to multivalued maps on a Cauchy complete metric space for developing a fixpoint semantics of  stratified disjunctive logic programs 
(see also~\cite{hitzler2010mathematical})\@. 
We are not aware, however, of previous attempts to develop fundamental concepts of system design based on multivalued fixpoints principles in the ultrametric space of streams.


The derivation of the fixpoint induction principle 
(Lemma~\ref{lem:fixpoint-principle-strong}) relies on 
fixpoint iteration.
This induction principle is reminiscent, of course,
of the fixpoint induction principle for 
Scott-continuous functions and Kleene's fixpoint theorem~\cite{winskel1993formal}\@. 
Moreover, the induction principle for multivalued stream 
transformers is analogous to Park's lemma, which is an immediate consequence of the Knaster-Tarski fixpoint theorem for monotone functions on complete lattices.


Vector-valued metrics with codomain $\Real^n$\@, for $n \geq 1$, are a viable alternative to the use of the supremum valuation for products of streams. 
The classical Banach contraction principle was extended by Perov~\cite{perov1964cauchy} for contraction mappings on spaces endowed with (finite dimensional) vector-valued metrics with codomain $\Real^n$\@.
Generalizations to multivalued contraction mappings 
have been developed, among others, by Filip and Perusel~\cite{filip2010fixed,almalki2022perov}\@.
These fixpoint theorems rely on contractive mappings.
Therefore, they are applicable to strongly causal but not to weakly causal stream transformers.
The advantage of taking the supremum valuation is that fixpoint results for multivalued mappings apply readily, whereas the use of vector-valued metrics requires explicit extensions of these results to this modified setting. 
On the other hand, contractions of endomorphisms with respect to vector-valued metric can be defined more generally
by requiring a square matrix $L$ with nonnegative entries such that $\lim_{k\to\infty} L^k = 0$\@. 
Here we restricted developments to the cases where $L$ is of the specific form $l \cdot I$ for $l < 1$ a Lipschitz constants and $I$ the identity matrix. 

More constructive proofs of fixpoints of deterministic transformers on spherically complete ultrametric spaces result in transfinite Picard iterations~\cite{khamsi2011introduction,kirk2014fixed}\@.
These kinds of results could therefore be used
to formulate induction principles for the solutions of deterministic causal stream transformers. But we are not aware of 
general results in this direction for 
nondeterministic maps.  





\section{Conclusions}\label{sec:conclusions}
Fundamental concepts of system design such as interfaces, composition, refinement, and abstraction
are defined and derived 
from the causality-contraction connection together with established results in fixpoint and approximation theory.
Moreover, 
the correspondence between the functional and a more machine-oriented view of causal stream transformers may advantageously be used in system design to switch between these intertwined view points as needed.
It is well beyond the scope and the ambition of this exposition, however, to develop a comprehensive and readily applicable framework for system design~\cite{broy2023}\@.

Quantifiable fixpoint approximation is one of the main practical advantages of the proposed metric-based approach to system design. 
These bounds provide useful {\em anytime information} both to human system designers and
to mechanized verification and design engines; such as: "the required 1024 element prefix of the stream solution has already been established in the current Picard iteration"\@.

In some cases, the causality condition on stream maps may be considered to be too  restrictive as it includes continuity.
This condition, however, can be relaxed since 
the ultrametric space of streams is 
{\em $\varepsilon$-chainable} (for given $ \varepsilon > 0$);
that is, there is a finite path between any two streams with intermediate "jumps" less than $\varepsilon$\@.  
For $\varepsilon$-chainable spaces a local 
contraction condition suffices to obtain
fixpoints for single-~\cite{edelstein1961extension} and
multivalued maps~\cite{nadler1969multi}\@.

The causality-contraction connection easily, but at the expense of notational overhead, extends to transformers with a mix of discrete and 
dense (that is, indices are real-valued) streams in heterogeneous products of streams to 
model {\em hybrid systems}\@. 
Moreover, {\em probabilistic systems} may be
modeled by Menger's probabilistic metric space, or any variant thereof, in which the distance between any two points is a probability distribution function~\cite{sherwood1971complete,kutbi2015further}\@. 
Together, these developments yield a comprehensive mathematical foundation
for the formal construction of cyber-physical systems~\cite{lee2008cyber,broy2013challenges} and their realization on computing machinery.

\bibliographystyle{ACM-Reference-Format}
\bibliography{main}


\end{document}